\documentclass[runningheads, envcountsect, 10pt]{llncs}

\usepackage{cite}
\usepackage{amsmath,amssymb,amsfonts, verbatim}
\usepackage{array}
\usepackage{algorithm}
\usepackage{algorithmicx}
\usepackage[noend]{algpseudocode}
\usepackage{graphicx}
\usepackage{textcomp}
\usepackage{xcolor}
\usepackage{framed}
\usepackage{url}
\usepackage{bm}
\usepackage[mathscr]{euscript}
\usepackage{subfigure, multirow} 
\usepackage{xspace}
\usepackage{enumerate}
\usepackage{booktabs}

\newcommand{\eat}[1]{}

\newcommand{\etc}{{\textrm{etc.}}\xspace}
\newcommand{\ie}{{\textrm{i.e.}}\xspace}

\newcommand{\etal}{{\textrm{et al.}}\xspace}

\newcommand{\Pproblem}{\texttt{P}\xspace}
\newcommand{\PTIME}{\texttt{PTIME}\xspace}
\newcommand{\NP}{\texttt{NP}\xspace}
\newcommand{\NC}{\texttt{NC}\xspace}
\newcommand{\DLOG}{\texttt{DLOGTIME}\xspace}

\newcommand{\PsTR}{\texttt{PsTR}\xspace}
\newcommand{\PsTE}{\texttt{PsTE}\xspace}
\newcommand{\PsT}{\texttt{PsT}\xspace}

\newcommand{\PsPL}{\texttt{PsPL}\xspace}

\newcommand{\PPL}{\texttt{PPL}\xspace}
\newcommand{\PT}{\texttt{PT}\xspace}

\newcommand{\RATM}{\texttt{RATM}\xspace}

\newcommand{\CP}{{\cal P}\xspace}
\newcommand{\Pre}{{\rm \Pi}\xspace}

\begin{document}
	
	\title{Sublinear-time Reductions for Big Data Computing}
	
	\author{Xiangyu Gao\inst{1,2} \and
		Jianzhong Li\inst{2, 1} \and
		Dongjing Miao\inst{1}\\
		\email{\{gaoxy, lijzh, miaodongjing\}@hit.edu.cn}
	}
	\institute{
		{Department of Computer Science and Technology, Harbin Institute of Technology, Harbin, China}\and
		{Faculty of Computer Science and Control Engineering, Shenzhen Institute of Advanced Technology Chinese Academy of Sciences, Shenzhen, China}
	}
	
	\titlerunning{Sublinear-time Reductions for Big Data Computing}
	\authorrunning{X. Gao, J. Li, et al.}
	
	\maketitle
	
	\begin{abstract}
		With the rapid popularization of big data, the dichotomy between tractable and intractable problems in big data computing has been shifted.
		Sublinear time, rather than polynomial time, has recently been regarded as the new standard of tractability in big data computing.
		This change brings the demand for new methodologies in computational complexity theory in the context of big data.
		Based on the prior work for sublinear-time complexity classes~\cite{DBLP:journals/tcs/GaoLML20}, this paper focuses on sublinear-time reductions specialized for problems in big data computing.
		First, the pseudo-sublinear-time reduction is proposed and the complexity classes \Pproblem and \PsT are proved to be closed under it.
		To establish \PsT-intractability for certain problems in \Pproblem, we find the first problem in $\Pproblem \setminus \PsT$.
		Using the pseudo-sublinear-time reduction, we prove that the nearest edge query is in \PsT but the algebraic equation root problem is not.
		Then, the pseudo-polylog-time reduction is introduced and the complexity class \PsPL is proved to be closed under it.
		The \PsT-completeness under it is regarded as an evidence that some problems can not be solved in polylogarithmic time after a polynomial-time preprocessing, unless \PsT = \PsPL.
		We prove that all \PsT-complete problems are also \Pproblem-complete, which gives a further direction for identifying \PsT-complete problems.
		\keywords{Big data computing, Sublinear-time tractability, Reduction techniques, Preprocessing}
	\end{abstract}

	\section{Introduction}

Traditionally, a problem is considered to be tractable if there exists a polynomial-time (\PTIME) algorithm for solving it.
However, \PTIME no more serves as a good yardstick for tractability in the context of big data, and sometimes even linear-time algorithms can be too slow in practice.
For example, a linear scan of a 1PB dataset with the fastest Solid State Drives on the market will take 34.7 hours~\cite{ssd.rank}.
Therefore, sublinear time is considered as the new standard of tractability in big data computing~\cite{lij2014dasfaa}.
This change has promoted the development of computational complexity theory specialized for problems in big data computing.

In the last few years, many complexity classes were proposed to formalize tractable problems in big data computing~\cite{DBLP:journals/pvldb/FanGN13, DBLP:journals/kais/YangWC18, DBLP:journals/tcs/GaoLML20}.
The first attempt was made by Fan \etal in 2013~\cite{DBLP:journals/pvldb/FanGN13}, which focuses on tractable boolean query classes with the help of preprocessing.
They defined a concept of $\sqcap$-tractability for boolean query classes.
A boolean query class is $\sqcap$-tractable if it can be processed in parallel polylogarithmic time (\NC) after a \PTIME preprocessing.
They defined a query complexity class $ {\sf\sqcap T^0_Q}$ to denote the set of $\sqcap$-tractable query classes.
To clarify the difference between ${\sf\sqcap T^0_Q}$ and \Pproblem, they proposed a form of generalized \NC reduction, referred as $F$-reduction $\le^{\sf\tt NC}_F$, and proved that $\sqcap {\sf T^0_Q}$ is closed under $F$-reduction.
They showed that $\NC \subseteq {\sf\sqcap T^0_Q} \subseteq \Pproblem$, but $\sqcap_{\sf T^0_Q} \ne \Pproblem$ unless $\Pproblem = \NC$.

Then, Yang \etal introduced a $\sqcap'$-tractability for short query classes, \ie the query length is bounded by a logarithmic function with respect to the data size~\cite{DBLP:journals/kais/YangWC18}.
On the basis of $\sqcap$-tractability theory, they placed a logarithmic-size restriction on the preprocessing result and relaxed the query execution time to polynomial.
The corresponding query complexity class was denoted as $\sqcap'{\sf T^0_Q}$, including the set of $\sqcap'$-tractable short query classes.
They proved that $F$-reduction is also compatible with $\sqcap'{\sf T^0_Q}$ and any $\sqcap'{\sf T^0_Q}$-complete query class under $F$-reduction is $P$-complete query class under \NC reduction.

A year ago, to completely describe the scope of sublinear-time tractable problems, the authors of this paper proposed two categories of sublinear-time complexity classes~\cite{DBLP:journals/tcs/GaoLML20}.
One kind characterizes the problems that are directly feasible in sublinear time, while the other describes the problems that are solvable in sublinear time after a \PTIME preprocessing.
However, we only showed that the polylogarithmic-time class \PPL is closed under \DLOG reduction and the sublinear-time class \PT is closed under linear-size \DLOG reduction, but left reductions for pseudo-sublinear-time complexity classes as a future work.

\noindent{\bf Open Question 1.} {\it What kind of reductions are appropriate for pseudo-sublinear-time tractable problems in big data computing?}

On the other, it is also important to identify the problems that are unsolvable in sublinear time.
Since, the new tractable standard in big data computing essentially dichotomizes problems in \Pproblem, it is significant to differentiate hardness of problems in \Pproblem.
The modern approach is to prove {\it conditional lower bounds} via {\it fine-grained reductions}~\cite{bringmann2019fine}.
Generally, a fine-grained reduction starts from a key problem such as SETH, 3SUM, APSP, \etc, which has a widely believed {\it conjecture} about its time complexity, and transfers the conjectured intractability to the reduced problem, yielding a conditional lower bounds on how fast the reduced problem can be solved.
The resulting area is referred as {\it fine-grained complexity theory}, and we refer to the surveys~\cite{DBLP:conf/iwpec/Williams15, williams2018some} for further reading.
However, to establish a problem is intractable in the context of big data, an unconditional lower bound, even rough, is also preferred.
Thus, the other goal of this paper is to overcome the following barrier.

\noindent{\bf Open Question 2.} {\it Is there a natural problem belonging to {\rm \Pproblem} but not to {\rm \PsT}?}

\subsection{Our Results}

The focus of this paper is mainly on pseudo-sublinear-time reductions specialized for problems in big data computing.
We reformulate the reduction used in~\cite{DBLP:journals/iandc/CadoliDLS02}, which was originally designed for complexity classes beyond \NP.
The general description of reductions proposed in this paper is illustrated in Figure~\ref{fig:red}.
We derive appropriate reductions for different complexity classes by limiting the computational power of functions used in it.

\begin{figure}[h]\vspace{-4ex}
	\centering
	\includegraphics[width=.65\linewidth]{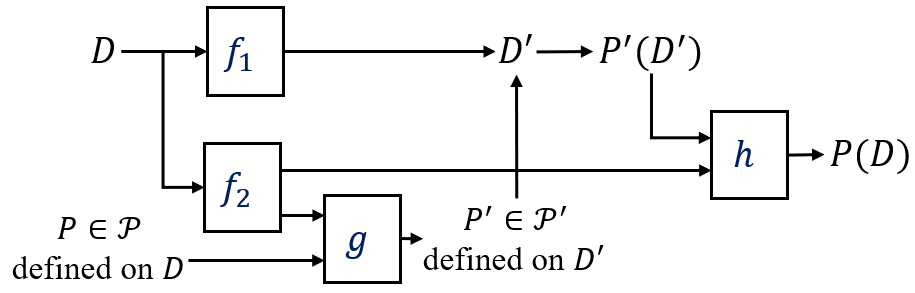}\vspace{-1ex}
	\caption{Illustration of reductions used in this paper.}\vspace{-4ex}
	\label{fig:red}
\end{figure}

We first introduce the pseudo-sublinear-time reduction, $\le^{\tt PsT}_m$ for problems in \PsT.
We prove that it is transitive and the complexity classes \Pproblem and \PsT are closed under $\le^{\tt PsT}_m$.
Due to the limitation of the fraction power function, we do not define a new \Pproblem-completeness under $\le^{\tt PsT}_m$ to include the problems in $\Pproblem\setminus \PsT$.
Instead, we prove a natural problem, the circuit value problem, can not be solved in sublinear time after a \PTIME preprocessing.
This also proves that $\PsT \subsetneq \Pproblem$.
After that, we reduce the algebraic equation root problem to the circuit value problem, which means the former also belongs to $\Pproblem \setminus \PsT$.
Moreover, we show the nearest neighbor problem is in \PsT by reducing it to the range successor query.

Then, we propose the notion of pseudo-polylog-time reduction, $\le^{\tt PsPL}_m$, and show that \PsPL is closed under $\le^{\tt PsPL}_m$.
We define the \PsT-completeness under $\le^{\tt PsPL}_m$, which can be treated as an evidence that certain problems are not solvable in polylogarithmic time after a \PTIME preprocessing unless \PsT = \PsPL.
We prove that all \PsT-complete problems are also \Pproblem-complete.
This specifies the range of possible \PsT-complete problems.

Moreover, we also extend L-reduction~\cite{DBLP:conf/coco/Crescenzi97} to pseudo-sublinear time and prove that it linearly preserve approximation ratio for pseudo-sublinear-time approximation algorithms.
Finally, we give a negative answer to the existence of complete problems in \PPL under \DLOG reduction.


\noindent{\bf Outline. }
The remainder of this paper is organized as follows.
Necessary preliminaries are stated in Section~\ref{sec:pre}.
The definitions and properties of pseudo-sublinear-time reduction and pseudo-polylog-time reduction are presented in Section~\ref{sec:ps-sub} and Section~\ref{sec:ps-polylog} respectively.
The pseudo-sublinear-time L-reduction is introduced in Section~\ref{sec:L}.
A negative results for the existence of complete problems in \PPL is shown in Section~\ref{sec:ppl}.
The paper is concluded in Section~\ref{sec:con}.
	\section{Preliminaries}\label{sec:pre}

In this section, we briefly review the sublinear-time complexity classes introduced in~\cite{DBLP:journals/tcs/GaoLML20} and the basic concepts of reductions.

We start with some notations.

\noindent{\bf Notations.}
To reflect the characteristics in big data computing, the input of a problem is partitioned into data part and problem part.
Thus, a decision problem $\CP$ can be considered as a binary relation such that for each $D$ and problem $P$ defined on $D$, $\langle D, P\rangle \in \CP$ if and only if $P(D)$ is true.
We say that a binary relation is in complexity class ${\cal C}$ if it is in ${\cal C}$ to decide whether a pair $\langle D, P\rangle \in \CP$.
Following the convention of complexity theory \cite{DBLP:books/daglib/0072413}, we assume a finite alphabet $\rm \Sigma$ of symbols to encode both of them.
The length of a string $x \in {\rm \Sigma}^*$ is denoted by $|x|$.
Given an integer $n$, let $\llcorner n \lrcorner$ denote the binary form of $n$.

\noindent{\bf Sublinear-time Complexity Classes.}
The computational model is crucial when describing sublinear-time computation procedures.
A random-access Turing machine (\RATM) $M$ is a $k$-tape Turing machine including a read-only input tape and $k - 1$ work tapes, referred as non-index tape.
And $M$ is additionally equipped with $k$ binary index tapes, one for each non-index tape.
$M$ has a special {\it random access} state which, when entered, moves the head of each non-index tape to the cell described by the respective index tape in one step.
Based on \RATM, a series of pure-sublinear-time complexity classes are proposed in \cite{DBLP:journals/tcs/GaoLML20} to include problems that are solvable in sublinear time.



\begin{definition}\label{def:ppl}
	\emph{
		The class \PPL consists of problems that can be solved by a \RATM in $O({\rm polylog}(n))$ time, where $n$ is the length of the input.
		And for each $i \ge 1$, $\PPL^i$ consists of problems that can be solved by a \RATM in $O(\log^i{n})$ time.
	}
\end{definition}

\begin{definition}\label{def:pt}
	\emph{
		The class \PT consists of problems that can be solved by a \RATM in $o(n)$ time, where $n$ is the length of the input.
	}
\end{definition}

Moreover, when the data part is fixed and known in advance, it makes sense to perform an off-line preprocessing on it to accelerate the subsequent processing of problem instances defined on it.
Hence, some pseudo-sublinear-time complexity classes are also defined to include the problems which are solvable in sublinear time after a \PTIME preprocessing on the data part.

\begin{definition}\label{def:pspl}
	\emph{
		A problem $\CP$ is in \PsPL if there exists a \PTIME preprocessing funciton $\Pre(\cdot)$ such that for any pair of strings $\langle D, P \rangle$ it holds that: $P(\Pre(D)) = P(D)$, and $P(\Pre(D))$ can be solved by a \RATM in $O({\rm polylog}(|D|))$ time.
	}
\end{definition}

\begin{definition}\label{def:pst}
	\emph{A problem $\CP$ is in \PsT if there exists a \PTIME preprocessing function $\Pre(\cdot)$ such that for any pair of strings $\langle D, P \rangle$ it holds that: $P(\Pre(D)) = P(D)$, and $P(\Pre(D))$ can be solved by a \RATM in $o(|D|)$ time.
	Moreover, a problem $\CP$ is in \PsTR (resp. \PsTE) if $\CP \in \PsT$ and the \PTIME preprocessing function $\Pre(\cdot)$ satisfies that for all big data $D$: $|\Pre(D)| < |D|$ (resp. $|\Pre(D)| \ge |D|$).
	}
\end{definition}



\noindent{\bf Reductions.}
In complexity theory, reductions are always used to both find efficient algorithms for problems, and to provide evidence that finding particularly efficient algorithms for some problems will likely be difficult\cite{DBLP:books/daglib/0066920, DBLP:books/daglib/0072413}.
Two main types of reductions are used in computational complexity theory, the many-one reduction and the Turing reduction. 
A problem $\CP_1$ is {\it Turing reducible} to a problem $\CP_2$, denoted as $\CP_1 \le_T \CP_2$ if there is an oracle machine to solve $\CP_1$ given an oracle for $\CP_2$.
That is, there is an algorithm for $\CP_1$ if it is available to a subroutine for solving $\CP_2$.
While, many-one reductions are a special case and stronger form of Turing reductions.
A decision problem $\CP_1$ is \textit{many-one reducible} to a decision problem $\CP_2$, denoted as $\CP_1 \le_m \CP_2$, if the oracle that is, the subroutine for $\CP_2$ can be only invoked once at the end, and the answer can not be modified.

Reductions define difficulty orders (from different aspects) among problems in a complexity class.
Hence, reductions are required to be transitive and easy to compute, relative to the complexity of typical problems in the class.
For example, when studying the complexity class \NP and harder classes such as the polynomial hierarchy, polynomial-time reductions are used, and when studying classes within \Pproblem such as \NC and {\sf\small NL}, log-space reductions are used.
We say a complexity class ${\cal C}$ is closed under a reduction if problem $\CP_1$ is reducible to another problem $\CP_2$ and if $\CP_2$ is in ${\cal C}$, then so must be $\CP_1$.
	\section{Pseudo-sublinear-time Reduction}\label{sec:ps-sub}

In this section, we introduce the notion of pseudo-sublinear-time reduction to tell whether a problem can be solved in sublinear time after a \PTIME preprocessing.

\begin{definition}\label{def:pseudo-sublinear-time}
	{\rm A decision problem $\CP_1$ is {\it pseudo-sublinear-time reducible} to a decision problem $\CP_2$, denoted as $\CP_1 \le^{\tt PsT}_m \CP_2$, if there is a triple $\langle f_1(\cdot), f_2(\cdot),  g(\cdot, \cdot)\rangle$, where $f_1(\cdot)$ and $f_2(\cdot)$ are linear-size \NC computable functions and $g(\cdot, \cdot)$ is a \PsT computable function, such that for any pair of strings $\langle D, P\rangle$ it holds that}
	\[\langle D, P \rangle \in \CP_1 \Leftrightarrow \langle f(D), g(D, P) \rangle \in \CP_2.\]
\end{definition}

Recall the general formalization of reductions specialized for problems in big data computing shown in Figure~\ref{fig:red}.
In contrast to traditional reductions such as polynomial-time reduction and log-space reduction, the pseudo-sublinear-time reduction is defined for the two parts of problems respectively.
Concretely speaking, (1) the data part of $\CP_2$ is obtained from the data part of $\CP_1$ using $f_1(\cdot)$, and (2) the problem part of $\CP_2$ is generated from the problem part of $\CP_1$ using $g(\cdot, \cdot)$ with some additional information of the data part of $\CP_1$ provided by $f_2(\cdot)$.
Intuitively, for different problems defined on the same data $D$, the computation of $f_2(D)$ can be regarded as an off-line process with a one-time cost.
Hence, when talking about the running time of $g(\cdot, \cdot)$, the running time of $f_2(\cdot)$ is excluded.
We first prove that $\le^{\tt PsT}_m$ is transitive.



\begin{theorem}\label{thm:ps-sub-trans}
	If $\CP_1 \le^{\tt PsT}_m \CP_2$ and $\CP_2 \le^{\tt PsT}_m \CP_3$, then also $\CP_1 \le^{\tt PsT}_m \CP_3$.
\end{theorem}
\begin{proof}
	From $\CP_1 \le^{\tt PsT}_m \CP_2$ and $\CP_2 \le^{\tt PsT}_m \CP_3$, it is known that there exist four linear-size \NC computable functions $f_1(\cdot)$, $f_2(\cdot)$, $f'_1(\cdot)$, and $f'_2(\cdot)$, and two \PsT computable functions $g(\cdot, \cdot), g'(\cdot,\cdot)$ such that for any pair of strings $\langle D_1, P_1\rangle$ and $\langle D_2, P_2\rangle$ it holds that
	\[\langle D_1, P_1 \rangle \in \CP_1 \Leftrightarrow \langle f_1(D_1), g(f_2(D_1), P_1) \rangle \in \CP_2,\]
	\[\langle D_2, P_2 \rangle \in \CP_2 \Leftrightarrow \langle f'_1(D_2), g'(f'_2(D_2), P_2) \rangle \in \CP_3.\]
	
	To show $\CP_1 \le^{\tt PsT}_m \CP_3$, we define three functions $f_1''(\cdot)$, $f''_2(\cdot)$ and $g''(\cdot)$ as follows.
	Let $f''_1(x) = f'_1(f_1(x))$, $f''_2(x) = \llcorner |f_2(x)|\lrcorner \# f_2(x)\#f'_2(f_1(x))$ and $g''(x, y) = g'(q, g(p, y))$ if $x = \llcorner |p| \lrcorner \# p \# q$, where \# is a special symbol that is not used anywhere else.
	Then we have
	\begin{equation*}
		\begin{aligned}
			\langle D_1, P_1 \rangle \in \CP_1 &\Leftrightarrow \langle f_1(D_1), g(f_2(D_1), P_1)\rangle \in \CP_2\\
			&\Leftrightarrow \langle f'_1(f_1(D_1)), g'(f'_2(f_1(D_1)), g(f_2(D_1), P_1)) \rangle \in \CP_3\\
			&\Leftrightarrow \langle f''_1(D_1), g''(\llcorner|f_2(D_1)|\lrcorner\#f_2(D_1)\#f'_2(f_1(D_1)), P_1) \rangle \in \CP_3\\
			&\Leftrightarrow \langle f''_1(D_1), g''(f''_2(D_1), P_1) \rangle \in \CP_3
		\end{aligned}
	\end{equation*}
	
	With the fact that the concentration and composition of two linear-size \NC computable function are still linear-size \NC computable functions, it is easy to verify that $f''_1(\cdot), f''_2(\cdot)$ are linear-size  \NC computable.
	As for $g''(\cdot, \cdot)$, the total time needed for computing $g''(f''_2(D), P)$ is bounded by $O(t_{g}(|f_2(D)|) + t_{g'}(|f'_2(f_1(D))|) + \log{|f_2(D)|}) = o(|D|)$.
	This completes the proof.\qed
\end{proof}

The pseudo-sublinear-time reduction is designed as a tool to prove that for some problems in \Pproblem, there is no algorithm can solve it in sublinear time after a \PTIME preprocessing.
Hence, in addition to time restriction, we also limit the output size of $f_1(\cdot)$ and $f_2(\cdot)$ to ensure that \PsT is closed under $\le^{\tt PsT}_m$.

\begin{theorem}\label{thm:ps-sub-closed}
	The complexity classes \Pproblem and \PsT is closed under $\le^{\tt PsT}_m$.
\end{theorem}

\begin{proof}
	To show \PsT is closed under $\le^{\tt PsT}_m$, we claim that for all $\CP_1$ and $\CP_2$ if $\CP_1 \le^{\tt PsT}_m \CP_2$ and $\CP_2 \in \PsT$, then $\CP_1 \in \PsT$.
	From $\CP_1 \le^{\tt PsT}_m \CP_2$, we know that there exist two linear-size \NC computable functions $f_1(\cdot)$ and $f_2(\cdot)$, and a \PsT computable function $g(\cdot, \cdot)$ such that for any pair of strings $\langle D_1, P_1 \rangle$ it holds that
	\[\langle D_1, P_1 \rangle \in \CP_1 \Leftrightarrow \langle f_1(D_1), g(f_2(D_1), P_1) \rangle \in \CP_2.\]
	
	Furthermore, since $\CP_2 \in \PsT$, there exists a \PTIME preprocessing function $\Pre_2(\cdot)$ such that for any pair of strings $\langle D_2, P_2\rangle$ it holds that: $P_2(\Pre_2(D_2)) = P_2(D_2)$, and $P_2(\Pre_2(D_2))$ can be solved by a \RATM $M_2$ in $o(|D_2|)$ time.
	Therefore, for any pair of strings $\langle D_1, P_1\rangle$ we have,
	\[P_1(D_1) = g(f_2(D_1), P_1)(f_1(D_1)) = g(f_2(D_1), P_1)(\Pre_2(f_1(D_1))).\]
	
	To show $\CP_1 \in \PsT$, we define a \PTIME preprocessing function $\Pre_1(\cdot)$ for $\CP_1$ such that $P_1(D_1) = P_1(\Pre_1(D_1))$ and a \RATM for $P_1(\Pre_1(D_1))$ running in sublinear time with respect to $|D_1|$.
	First, let $\Pre_1(x) = \llcorner \left|f_2(x)\right| \lrcorner\#f_2(x)\#\Pre_2(f_1(x))$, where \# is a special symbol that is not used anywhere else.
	It is remarkable to see that $\llcorner \left|f_2(x)\right| \lrcorner$ is used to help us to distinguish the two parts of the input in logarithmic time.
	Then we construct a \RATM $M_1$ by appending a pre-procedure to $M_2$.
	More concretely, with input $\Pre_1(D_1)$ and $P_1$, $M_1$ first copies $\llcorner |f_2(D_1)| \lrcorner$ to its work tap and computes the index of the second \#, which equals to $|f_2(D_1)| + |\llcorner |f_2(D_1)| \lrcorner| + 1$.
	Then $M_1$ generates $g(f_2(D_1), P_1)$ according to the information between the two \#s.
	Finally, $M_1$ simulates the computation of $M_2$ with input $\Pre_2(f_1(D_1))$, the information behind the second \#, and $g(f_2(D_1), P_1)$, then outputs the result returned by $M_2$.
	
	Since $\Pre_2(\cdot)$ is \PTIME computable, both $f_1(\cdot)$ and $f_2(\cdot)$ are $\NC$ computable, and the length of a string is logarithmic time computable, the running time of $\Pre_1(\cdot)$ can bounded by a polynomial.
	The time required by computing the index of the second \# is $t_I = O(\log{|f_2(D_1)|})$
	And, $g(\cdot, \cdot)$ is computable in $o(|f_2(D_1)|)$ time.
	As both $f_1(\cdot)$ and $f_2(\cdot)$ are linear-size functions, the running time of $M_1$ is bounded by $t_I + t_g + t_{M_2} = O(\log{|f_2(D_1)|} + o(|f_2(D_1)|) + o(|f_1(D_1)|) = o(|D_1|)$.
	Thus, $\CP_1 \in \PsT$.
	
	As for \Pproblem, we can consider another characterization for problems in \Pproblem.
	That is, there is a \PTIME preprocessing function $\Pre(\cdot)$ and a \PTIME \RATM $M$ such that for any pair of strings $\langle D, P\rangle$ it holds that: $P(\Pre(D)) = P(D)$ and $P(\Pre(D))$ can be solved by $M$.
	Then with similar construction as above, it is easy to prove that \Pproblem is closed under $\le^{\tt PsT}_m$.\qed
\end{proof}

The reduction defines a partial order of computational difficulty of problems in a complexity class, and the complete problems are regarded as the hardest ones.
Analogous to \NP-completeness, the \Pproblem-complete problems under $\le^{\tt PsT}_m$ can be considered as intractable problems in $\Pproblem\setminus \PsT$ if $\Pproblem \ne \PsT$.
However, we don't think it is appropriate to define that new \Pproblem-completeness for the following reason.
According to the proofs of the first complete problem of \Pproblem (under \NC reduction) and \NP, we notice that the size of the resulted instance is always related to the running time of the Turing machine for the origin problem.
Hence, the linear-size restriction of $f_1(\cdot)$ and $f_2(\cdot)$ may be too strict to hold.
Nevertheless, we succeeded to find a natural problem in $\Pproblem \setminus \PsT$.
Then, based on it, we can establish the unconditional pseudo-sublinear-time intractability for problems in $\Pproblem\setminus \PsT$.\smallskip

\noindent {\bf Circuit Value Problem (CVP):}\vspace{-1.5ex}
\begin{itemize}
	\item[$\circ$]{\bf Given:} A Boolean circuit $\alpha$, and inputs $x_1, \cdots, x_d$.
	\item[$\circ$]{\bf Problem:} Is the output of $\alpha$ is {\tt TRUE} on inputs $x_1, \cdots, x_d$? 
\end{itemize}

\begin{theorem}\label{thm:cvp}
	There is no algorithm can preprocess a circuit $\alpha$ in polynomial time and subsequently answer whether the output of $\alpha$ on the input $x_1, \cdots, x_d$ is {\tt TRUE} in sublinear time.
	That is, {\rm CVP} $\in$ {\rm \Pproblem} $\setminus$ {\rm \PsT}.
\end{theorem}
\begin{proof}
	As stated in~\cite{HOLDSWORTH200243}, given $d$ variables, there are ${2^2}^d$ distinct boolean functions can be constructed in total.
	And each of them can be written as a full disjunctive normal from its truth table, which can easily represented by a circuit.
	Suppose CVP belongs to \PsT, \ie, there is a \PTIME preprocessing function $\Pre(\cdot)$ on $\alpha$ such that for all interpretations of $x_1, \cdots, x_d$, $\alpha(x_1, \cdots, x_d) = \Pre(\alpha)(x_1, \cdots, x_d)$ can be computed in sublinear time with respect to $|\alpha|$.
	Consider any two distinct circuits $\alpha_1$ and $\alpha_2$ with the same variables $x_1, \cdots, x_d$.
	There exists an interpretation for $x_1, \cdots, x_d$ such that $\alpha_1(x_1, \cdots, x_d) \ne \alpha_2(x_1, \cdots, x_d)$.
	Consequently, $\Pre(\alpha_1) \ne \Pre(\alpha_2)$.
	Therefore, all these circuits have different outputs of the function $\Pre(\cdot)$.
	Since there are totally ${2^2}^d$ different circuits, then there should be at least ${2^2}^d$ different outputs of $\Pre(\cdot)$ on all these circuits.
	To denote these, the length of $\Pre(\alpha)$ should be at least $\log{2^2}^d = 2^d$.
	This contradicts to $\Pre(\cdot)$ is \PTIME computable by choosing $d = \omega(\log{|\alpha|})$.
	\qed
\end{proof}

\noindent {\bf Algebraic Equation Root Problem(AERP):}\vspace{-1.5ex}
\begin{itemize}
	\item[$\circ$]{\bf Given:} An algebraic equation $P$ with variables $x_1,\cdots, x_d$, and an assignment $A = (a_1, \cdots, a_d)$.
	\item[$\circ$]{\bf Problem:} Is $A$ a root of $P$?
\end{itemize}

\begin{theorem}
	${\rm CVP} \le^{\tt PsT}_m {\rm AERP}$.
\end{theorem}
\begin{proof}
	Assume we are given a boolean circuit $\alpha$, we define a transformation of $\alpha$ into an equation $P$ such that the output of $\alpha$ is {\tt\small TRUE} on inputs $x_1, \cdots, x_d$ if and only if $A = (x_1, \cdots, x_n)$ is a root of $P$.
	First, let $f_1(\cdot)$, $f_2(\cdot)$ express the following procedure.
	Traverse $\alpha$ in a topological order: (1) if an {\tt\small AND} gate with input $u$, $v$ is met, represent it by $u \times v$, (2) if an {\tt\small OR} gate with input $u$, $v$ is met, represent it by $u + v - u \times v$, (3) if a {\tt\small NOT} gate with input $u$, is met, represent it by $1 - u$, (4) if the final output gate $z$ is met, represent it by $z = 1$.
	Then, for each $x_i$ if the input $x_i$ is {\tt\small TRUE}, $g(f_2(\alpha), x_i) = 1$, otherwise, $g(f_2(\alpha), x_i) = 0$. 
	
	It is easy to see that the output of $\alpha$ is {\tt\small TRUE} on inputs $x_1, \cdots, x_d$ if and only if $A = (g(f_2(\alpha), x_1), \cdots, g(f_2(\alpha), x_n))$ is a root of $f_1(\alpha)$.
	And as stated in~\cite{cook1985taxonomy}, the topological traversal of a DAG can be computed in \NC.
	Moreover, both $|f_1(\alpha)|$ and $|f_2(\alpha)|$ are less than $7|\alpha|$.
	And let $d = o(|\alpha|)$, $g(\cdot, \cdot)$ is \PsT computable.\qed
\end{proof}

\begin{corollary}
	There is no algorithm can preprocess an algebraic equation $P$ in  polynomial time and subsequently answer whether a given assignment $A$ is a root of $P$ in sublinear time.
\end{corollary}

\eat{
\noindent Online Set Disjointness:\vspace{-1ex}
\begin{itemize}
	\item[$\circ$]{\bf Given:} a family of sets ${\cal F}$, and two sets $S, S' \in {\cal F}$.
	\item[$\circ$]{\bf Problem:} Are $S$ and $S'$ disjoint or not?
\end{itemize}

\noindent Common Neighbors Detection:\vspace{-1ex}
\begin{itemize}
	\item[$\circ$]{\bf Given:} a graph $G$, and two vertices $u$ and $v$.
	\item[$\circ$]{\bf Problem:} Does $u$ and $v$ share a common neighbor.
\end{itemize}

$d$-dimensional NNS can be solved with the help of k-d tree in $O(dn^{1-\frac{1}{d}})$ time.
}

Also, $\le^{\tt PsT}_m$ can also be used to derive efficient algorithms for problems in $\PsT$.
In the breakthrough work of dynamic DFS on undirected graphs~\cite{DBLP:conf/soda/BaswanaCC016}, Baswana \etal defined a nearest edge query between a subtree and an ancestor-descendant path in the procedure of rerooting a DFS tree, which was used in almost all subsequent work.
Chen et al. showed that this query could be solved by running a range successor query~\cite{DBLP:conf/swat/ChenDWZZ18}.
We refine the procedure as a pseudo-sublinear-time reduction. 
The definitions of these two problems are given as follows.\smallskip

\noindent{\bf Nearest Edge Query (NEQ):}\vspace{-1.5ex}
\begin{itemize}
	\item[$\circ$]{\bf Given:} A DFS tree $T$ of graph $G$, the endpoints $x, y$ of an ancestor-descendant path, the root $w$ of a subtree $T(w)$ such that $par(w) \in path(x, y)$.
	\item[$\circ$]{\bf Problem:} Find the edge $e$ that is incident nearest to $x$ among all edges between $T(w)$ and $path(x, y)$.
\end{itemize}

\noindent{\bf Range Successor Query (RSQ):}\vspace{-1.5ex}
\begin{itemize}
	\item[$\circ$]{\bf Given:} A set of $d$-dimensional points $S$, a query rectangle $Q = {\rm \Pi}^d_{i = 1}[a_i, b_i]$.
	\item[$\circ$]{\bf Problem:} Find the point $p$ with smallest $x$-coordinate among all points that are in the rectangle $Q$.
\end{itemize}

\begin{theorem}{\rm \cite{DBLP:conf/swat/ChenDWZZ18}}
	${\rm NEQ} \le_m^{\tt PsT} {\rm RSQ}$.
\end{theorem}

\begin{proof}
	Given a graph $G = (V, E)$ and a DFS tree $T$ of $G$, define $f_1:E \to S$ as follows, where $S$ is a set of $2$-dimensional points.
	Denote the preorder traversal sequence of $T$ by $\rho$, note that every subtree of $T$ can be represented by a continuous interval of $\rho$.
	Let $\rho(v)$ denote the index of vertex $v$ in this sequence that is if $v$ is the $i$-th element in $\rho$, then $\rho(v) = i$.
	For each edge $(u, v) \in E$, $f_1((u, v)) = (\rho(u), \rho(v))$.
	That is for each edge $(u, v)$, a point $(\rho(u), \rho(v))$ is added into $S$.
	Notice that for each point $p \in S$, there exists exactly one edge $(u, v)$ associated with $p$.
	Next we state the information provided by $f_2(\cdot)$.
	For each vertex $v$, let $\gamma(v) = \max_{w \in T(v)}\rho(w)$, \ie, the maximum index of vertices in $T(v)$.
	Thus, define $f_2(v)$ as $\rho(v)\#\gamma(v)$ for each $v \in V$.
	
	Then, to answer an arbitrary query instance $T(w), p(x, y)$, let $g$ be the function mapping $w, x, y$ to a rectangles ${\rm\Omega} = [\rho(x),\rho(w) - 1]\times[\rho(w), \gamma(w)]$.
	Finally, given a point $p \in S$ as the final result of RSQ, let $h(\cdot, \cdot)$ be reverse function of $f_1(\cdot)$, \ie, it returns the edge of $G$ corresponding to $p$.
	It is easily to verify that the edge corresponding to the point with minimum x-coordinate is the edge nearest to $x$ among all edges between $T(w)$ and $path(x, y)$ \cite{DBLP:conf/swat/ChenDWZZ18}.
	
	The preorder traversal sequence of $T$ can be obtained by performing a DFS on it, which can be done in \NC as stated in \cite{smith1986parallel}.
	Therefore, both $f_1(\cdot)$ and $f_2(\cdot)$ are \NC computable.
	Moreover, since each $e \in E$, there is a point $p = f_1(e)$ in $S$ and for each point $p \in S$, there is exactly one edge $e$ associated with $p$, we have $|f_1(G)| \in O(|G|)$.
	Similarly, for each vertex $v$, $f_2(v)$ records two values for it.
	Hence, $|f_2(G)| \in O(|G|)$.
	As for $g(\cdot, \cdot)$ and $h(\cdot, \cdot)$, with the mapping provided by $f_2(\cdot)$, both of them can be computed in sublinear time.
	\qed
\end{proof}

Notice that for optimization problems, we need not only the functions converting the data part and problem part of $\CP_1$ to corresponding part of $\CP_2$, but also a function $h(\cdot, \cdot)$ mapping the solution of $\CP_2$ back to the solution of $\CP_1$.
The resources restriction of $h(\cdot, \cdot)$ is set to be the same as $g(\cdot, \cdot)$.
There is numerous work showing that RSQ belongs to \PsT \cite{DBLP:conf/swat/NekrichN12}.
Hence, with the fact that the complexity class \PsT is closed under $\le^{\tt PsT}_m$, the following corollary is obtained.

\begin{corollary}
	${\rm NEQ} \in $ {\rm \PsT}.
\end{corollary}

	\section{Pseudo-polylog-time Reduction}\label{sec:ps-polylog}

In this section, we introduce the notion of pseudo-polylog-time reduction, which will be used to clarify the difference between \PsT and \PsPL.

\begin{definition}\label{def:red-sps-plog}
	\emph{
		A decision problem $\CP_1$ is \textit{pseudo-polylog-time reducible} to a decision problem $\CP_2$, denoted as $\CP_1 \le^{\tt PsPL}_m \CP_2$, if there is a triple $\langle f_1(\cdot), f_2(\cdot), g(\cdot, \cdot)\rangle$, where $f_1(\cdot)$ and $f_2(\cdot)$ are \NC computable functions and $g(\cdot, \cdot)$ is a \PPL computable function, such that for any pair of strings $\langle D, P\rangle$ it holds that
		\[\langle D, P\rangle \in \CP_1 \Leftrightarrow \langle f_1(D), g(f_2(D), P) \rangle \in \CP_2.\]
	}
\end{definition}

With similar proof of Theorem \ref{thm:ps-sub-trans} and Theorem \ref{thm:ps-sub-closed}, we can show that $\le^{\tt PsPL}_m$ is transitive and the complexity class \PsPL is closed under $\le^{\tt PsPL}_m$..

\begin{theorem}
	If $\CP_1 \le^{\tt PsPL}_m \CP_2$ and $\CP_2 \le^{\tt PsPL}_m \CP_3$, then also $\CP_1 \le^{\tt PsPL}_m \CP_3$.
\end{theorem}

\begin{proof}
	From $\CP_1 \le^{\tt PsPL}_m \CP_2$ and $\CP_2 \le^{\tt PsPL}_m \CP_3$, it is known that there exist four \NC computable functions $f_1(\cdot), f'_1(\cdot)$, $f_2(\cdot)$ $f'_2(\cdot)$, and two \PPL computable functions $g(\cdot, \cdot), g'(\cdot,\cdot)$ such that for any pair of strings $\langle D_1, P_1\rangle$ and $\langle D_2, P_2\rangle$ it holds that
	\[\langle D_1, P_1 \rangle \in \CP_1 \Leftrightarrow \langle f_1(D_1), g(f_2(D_1), P_1) \rangle \in \CP_2,\]
	\[\langle D_2, P_2 \rangle \in \CP_2 \Leftrightarrow \langle f'_1(D_2), g'(f'_2(D_2), P_2) \rangle \in \CP_3.\]
	
	To show $\CP_1 \le^{\tt PsPL}_m \CP_3$, we define two \NC computable functions $f_1''(\cdot)$, $f''_2(\cdot)$ and a \PPL computable function $g''(\cdot)$ as follows.
	Let $f''_1(x) = f'_1(f_1(x))$, $f''_2(x) = \llcorner |f_2(x)|\lrcorner \# f_2(x)\#f'_2(f_1(x))$ and $g''(x, y) = g'(p, g(q, y))$ if $x = \llcorner |p| \lrcorner \# p \# q$, where \# is a special that is not used anywhere else.
	Then we have
	\begin{equation*}
		\begin{aligned}
			\langle D_1, P_1 \rangle \in \CP_1 &\Leftrightarrow \langle f_1(D_1), g(f_2(D_1), P_1)\rangle \in \CP_2\\
			&\Leftrightarrow \langle f'_1(f_1(D_1)), g'(f'_2(f_1(D_1)), g(f_2(D_1), P_1)) \rangle \in \CP_3\\
			&\Leftrightarrow \langle f''_1(D_1), g''(\llcorner|f_2(D_1)|\lrcorner\#f_2(D_1)\#f'_2(f_1(D_1)), P_1) \rangle \in \CP_3\\
			&\Leftrightarrow \langle f''_1(D_1), g''(f''_2(D_1), P_1) \rangle \in \CP_3
		\end{aligned}
	\end{equation*}
	It is easy to verify that $f''_1(\cdot), f''_2(\cdot)$ are in \NC and $g''(\cdot, \cdot)$ is in \PPL.
	\qed
\end{proof}

\begin{theorem}\label{thm:ps-polylog-closed}
	The complexity class \PsPL is closed under $\le^{\tt PsPL}_m$.
\end{theorem}

\begin{proof}
	From $\CP_1 \le^{\tt PsPL}_m \CP_2$, we know that there exist two \NC computable functions $f_1(\cdot)$, $f_2(\cdot)$, and a \PPL computable function $g(\cdot, \cdot)$ such that for any pair of strings $\langle D_1, P_1\rangle$ it holds that
	\[\langle D_1, P_1 \rangle \in \CP_1 \Leftrightarrow \langle f_1(D_1), g(f_2(D_1), P_1) \rangle \in \CP_2.\]
	
	Furthermore, since $\CP_2 \in \PsPL$, there exists a \PTIME preprocessing function $\Pre_2(\cdot)$ such that for any pair of strings $\langle D_2, P_2\rangle$ it holds that: $P_2(\Pre_2(D_2)) = P_2(D_2)$, and $P_2(\Pre_2(D_2))$ can be solved by a \RATM $M_2$ in $O(\log^{c_2}{|D_2|})$ for some $c_2 \ge 1$.
	Therefore, for any pair of strings $\langle D_1, P_1\rangle$ we have,
	\[P_1(D_1) = g(f_2(D_1), P_1)(f_1(D_1)) = g(f_2(D_1), P_1)({\rm\Pi_2}(f_1(D_1))).\]
	
	To show $\CP_1 \in \PsPL$, we claim that there exist a \PTIME preprocessing function ${\rm \Pi_1}(\cdot)$ for $\CP_1$ such that $P_1(D_1) = P_1({\rm\Pi_1}(D_1))$ and a \RATM for $P_1({\rm\Pi}_1(D_1))$ running in polylogarithmic time as required in Definition \ref{def:pspl}.
	First, let ${\rm\Pi_1}(x) = \llcorner|f_2(x)|\lrcorner\#f_2(x)\#{\rm\Pi_2}(f_1(x))$, where $\#$ is a special symbol that is not used anywhere else.
	Then we construct a \RATM $M_1$ by appending a pre-procedure to $M_2$.
	More concretely, with input ${\rm\Pi_1}(D_1)$ and $P_1$, $M_1$ first copies $\llcorner|f_2(x)|\lrcorner$ to one of its work tapes and computes the index of the second $\#$, which equals to $|f_2(x)| + |\llcorner|f_2(x)|\lrcorner| + 1$.
	Then $M_1$ generates $g(f_2(D_1), P_1)$ according to the information between the two \#s.
	Finally, $M_1$ simulates the computation of $M_2$ with input ${\rm\Pi}_2(f_1(D_1))$ behind the second \# and $g(f_2(D_1), P_1)$, then outputs the result returned by $M_2$.
	
	Since ${\rm\Pi_2}(\cdot)$ is in \PTIME, $f_1(\cdot)$ and $f_2(\cdot)$ are in \NC, and the length of a string is logarithmic time computable ${\rm\Pi_1}(\cdot)$ is obviously in \PTIME.
	Notice that computing the index of the second \# requires $t_I = O(\log{|f_2(D_1)|})$ time and $g(\cdot, \cdot)$ is computable in time $O(\log^{c_3}{|f_2(D_1)|})$ for some $c_3 \ge 1$.
	Therefore, the total running time of $M_1$ is bounded by $t_I + t_g + t_{M_2} = O(\log^{c_3}{|f_2(D_1)|} + \log^{c_2}{|f_1(D_1)|}) = O(\log^{c_1}{|D_1|})$ where $c_1 = \max\{c_2, c_3\}$.
	Thus, $\CP_1 \in \PsPL$.
	\qed
\end{proof}

Due to the limitations of fractional power functions, the complexity class \PsT is not closed under $\le^{\tt PsPL}_m$ unless we add an addition linear-size restriction of function $f_1(\cdot)$.
Fortunately, this does not prevent us from defining \PsT-completeness.

\begin{definition}
	\emph{
		A problem \CP is \PsT-hard under $\le^{\tt PsPL}_m$ if $\CP' \le^{\tt PsPL}_m \CP$ for all $\CP' \in \PsT$.
		A problem \CP is \PsT-complete under $\le^{\tt PsPL}_m$ if \CP is \PsT-hard and $\CP \in \PsT$.
	}
\end{definition}

Identifying the \PsT-complete problems may help us to separate \PsT and \PsPL.
That is if there is a \PsT-complete problem belonging to \PsPL, then $\PsPL = \PsT$.
In the following, we give a specified range of possible complete problems for \PsT, by relating them to a well-known \Pproblem-complete problem.
Given a graph $G$, a depth-first search(DFS) traverses $G$ in a particular order by picking an unvisited vertex $v$ from the neighbors of the most recently visited vertex $u$ to search, and backtracks to the vertex from where it came when a vertex $u$ has explored all possible ways to search further.\smallskip

\noindent {\bf Ordered Depth-First Search (ODFS):}\vspace{-1.5ex}
\begin{itemize}
	\item[$\circ$]{\bf Given:} A graph $G = (V, E)$ with fixed adjacent lists, fixed starting vertex $s$, and vertices $u$ and $v$.
	\item[$\circ$]{\bf Problem:} Does vertex $u$ get visited before vertex $v$ in the DFS traversal of $G$ starting from $s$? 
\end{itemize}

\begin{theorem}{\rm \cite{DBLP:journals/ipl/Reif85}}\label{odfs:pcomplete}
	{\rm ODFS} is {\rm \Pproblem}-complete under {\rm\NC} reduction.
\end{theorem}

\begin{theorem}
	Given a problem $\CP$, if $\CP$ is {\rm \PsT}-complete, then $\CP$ is {\rm \Pproblem}-complete.
\end{theorem}

\begin{proof}
	It is easy to see that ODFS is in \PsT.
	Since $\CP$ is \PsT-complete, ${\rm ODFS} \le^{\tt PsPL}_m \CP$.
	That is, there exist two \NC computable functions $f_1(\cdot)$, $f_2(\cdot)$ and a \PPL computable function $g(\cdot,\cdot)$ such that for all $\langle [G, s], [u, v] \rangle$ it holds that
	\[\langle [G, s], [u, v] \rangle \in {\rm ODFS} \Leftrightarrow \langle f_1([G, s]), g(f_2([G, s]), [u, v])\rangle \in \CP.\]
	As stated in Theorem \ref{odfs:pcomplete}, ODFS is \Pproblem-complete under \NC reduction.
	For any problem $L \in \Pproblem$, there is a \NC computable function $h(\cdot)$ such that
	\[x \in L \Leftrightarrow h(x) \in {\rm ODFS}.\]
	Recall that the input of ODFS consists of a graph $G$, a starting point $s$, and two vertices $u, v$.
	It is easy to modify the output format of $h(x)$ to $\llcorner|[G, s]|\lrcorner\#[G, s]\#[u,v]$ in \NC, where \# is a new symbol that is not used anywhere else.
	Now let $f'_1(x) = f_1(y)$ and $g'(x) = g(f_2(y), z)$, if $x = \llcorner|y|\lrcorner\#y\#z$.
	The two separators \# can be founded in logarithmic time.
	Consequently, it follows that
	\[x \in L \Leftrightarrow \langle h(x).[G,  s], h(x).[u, v]\rangle \in {\rm ODFS} \Leftrightarrow \langle f'_1(h(x)), g'( h(x))\rangle \in \CP.\]
	Let $h'(x) = f'_1(h(x))\circ g'(h(x))$ to denote the concentration of two parts of $\CP$ we can see that $L$ is \NC reducible to $\CP$.
	Therefore, $\CP$ is \Pproblem-complete.
	\qed
\end{proof}

	\section{Approximation Preserving Pseudo-sublinear-time Reduction}\label{sec:L}
A natural approach to cope with problems in $\Pproblem \setminus \PsT$ or that are \PsT-complete is to design pseudo-sublinear-time approximation algorithm.
Hence, in this section, we propose the pseudo-sublinear-time L-reduction, and prove that it linearly preserves approximation ratio for pseudo-sublinear-time approximation algorithms.

Let $\CP$ be a big data optimization problem, given a dataset $D$ and a problem instance $P \in \CP$ defined on $D$, let $P(D)$ denote the set of feasible solutions of $P$, and for any feasible solution $y \in P(D)$, let $\tau_\CP(y)$ denote the positive measure of $y$, which is called the objective function.
The goal of an optimization problem with respect to a problem instance $P \in \mathcal{P}$ is to find an optimum solution, that is, a feasible solution $y$ such that $\tau_{\cal P}(y) = \{\max, \min\}\{\tau_{\CP}(y') : y'\in P(D)\}$.
In the following, $\mathsf{opt}_{\mathcal{P}}$ will denote the function mapping an instance $P \in \mathcal{P}$ defined on $D$ to the measure of an optimum solution.

What's more, for each feasible solution $y$ of $D, P$, the \textit{approximation ratio} of $y$ with respect to $D, P$ is defined as $\rho(D, P, y) = \max\left\{\frac{\tau_\CP(y)}{{\tt opt}_{\CP}(D, P)}, \frac{{\tt opt}_{\CP}(D, P)}{\tau_\CP(y)}\right\}$.
The approximation ratio is always a number greater than or equal to 1 and is as close to 1 as the value of the feasible solution is close to the optimum value.
Let ${\cal A}$ be an algorithm that for any $D$ and problem instance $P \in \CP$ defined on $D$, returns a feasible solution ${\cal A}(\Pre(D), P)$ in sublinear time after a \PTIME preprocessing $\Pre(\cdot)$.
Given a rational $r \ge 1$, we say that ${\cal A}$ is an $r$-approximation algorithm for \CP if the approximation ratio of the feasible solution ${\cal A}(\Pre(D), P)$ with respect to $D, P$ satisfies $\rho_{\cal A}(D, P, {\cal A}(\Pre(D), P)) \le r$.

\begin{definition}
	\emph{
		A problem $\CP_1$ is \textit{pseudo-polylog-time L-reducible} to a problem $\CP_2$, denoted as $\mathcal{P}_1 \le^{\tt PsPL}_L \mathcal{P}_2$, if there is a pseudo-polylog-time reduction $\langle f_1(\cdot), f_2(\cdot), g(\cdot, \cdot), h(\cdot, \cdot)\rangle$ from $\CP_1$ to $\CP_2$ such that for all $D$ and $P \in \CP_1$ defined on $D$ it holds that:
		\begin{enumerate}[1.]
			\item $\mathsf{opt}_{\mathcal{P}_2}(f_1(D), g(f_2(D), P)) \le \alpha \cdot \mathsf{opt}_{\mathcal{P}_1}(D, P)$
			\item for any $y \in \mathsf{sol}_{\mathcal{P}_2}(f_1(D), g(f_2(D), P))$,
			\[|\mathsf{opt}_{\mathcal{P}_1}(D, P) - \tau_{\CP_1}(h(f_2(D), y))| \le \beta \cdot |\mathsf{opt}_{\mathcal{P}_2}(f_1(D), g(f_2(D), P)) - \tau_{\CP_2}(y)|.\]
		\end{enumerate}
	}
\end{definition}

\begin{theorem}\label{thm:ps-L}
	Given two problems $\CP_1$ and $\CP_2$, if $\CP_1 \le_L^\emph{\tt PsPL} \CP_2$ with parameter $\alpha$ and $\beta$
	and there is a pseudo-polylog-time $(1+\delta)$-approximation algorithm for $\CP_2$, then there is a pseudo-polylog-time $(1+\gamma)$-approximation algorithm for $\CP_1$, where $\gamma = \alpha\beta\cdot \delta$ if $\CP_1$ is a minimization problem and and $\gamma = \frac{\alpha\beta\delta}{1-\alpha\beta\delta}$ if $\CP_1$ is a maximization problem.
\end{theorem}

\begin{proof}
	The algorithm for $\CP_1$ is constructed as stated in the proof of Theorem \ref{thm:ps-polylog-closed}.
	Then, if $\CP_1$ is a minimization problem, it holds that
	\begin{equation*}
		\begin{small}
			\begin{aligned}
				\frac{\tau_{\CP_1}(h(D, P, y))}{{\tt opt}_{\CP_1}(D, P)}
				&= \frac{{\tt opt}_{\CP_1}(D, P) + \tau_{\CP_1}(h(D, P, y)) - {\tt opt}_{\CP_1}(D, P)}{{\tt opt}_{\CP_1}(D, P)}\\
				& \le \frac{{\tt opt}_{\CP_1}(D, P) + \beta\cdot \left|\tau_{\CP_2}(y) - {\tt opt}_{\CP_2}(f(D), g(D, P))\right|}{{\tt opt}_{\CP_1}(D, P)}\\
				& \le 1 + \alpha\beta \cdot \left|\frac{\tau_{\CP_2}(y) - {\tt opt}_{\CP_2}(f(D), g(D, P)))}{{\tt opt}_{\CP_2}(f(D), g(D, P))}\right|
			\end{aligned}
		\end{small}
	\end{equation*}
	Thus we obtain a $(1 + \alpha\beta \cdot \delta)$-approximation algorithm for $\CP_1$.
	And, if $\CP_1$ is a maximization problem, it holds that
	\begin{equation*}
		\begin{small}
			\begin{aligned}
				\frac{\tau_{\CP_1}(h(D, P, y))}{{\tt opt}_{\CP_1}(D, P)}
				&= \frac{{\tt opt}_{\CP_1}(D, P) + \tau_{\CP_1}(h(D, P, y)) - {\tt opt}_{\CP_1}(D, P)}{{\tt opt}_{\CP_1}(D, P)}\\
				&\ge \frac{{\tt opt}_{\CP_1}(D, P) - \beta\cdot \left|{\tt opt}_{\CP_2}(f(D), g(D, P)) - \tau_{\CP_2}(y)\right|}{{\tt opt}_{\CP_1}(D, P)}\\
				&\ge 1 - \alpha\beta \cdot \left|\frac{{\tt opt}_{\CP_2}(f(D), g(D, P)) - \tau_{\CP_2}(y))}{{\tt opt}_{\CP_2}(f(D), g(D, P))}\right|
			\end{aligned}
		\end{small}
	\end{equation*}
	Thus the algorithm is a $(1 + \frac{\alpha\beta\delta}{1 - \alpha\beta\delta})$-approximation algorithm for $\CP_1$.
	\qed
\end{proof}

It is easy to extend the above definition in the context of pseudo-sublinear-time reduction.
Hence, the following theorem is derived.

\begin{theorem}\label{thm:sps-L}
	Given two problems $\CP_1$ and $\CP_2$, if $\CP_1 \le^{\tt PsT}_L \CP_2$ with parameter $\alpha$ and $\beta$
	and there is a pseudo-sublinear-time $(1+\delta)$-approximation algorithm for $\CP_2$, then there is a pseudo-sublinear-time $(1+\gamma)$-approximation algorithm for $\CP_1$, where $\gamma = \alpha\beta\cdot \delta$ if $\CP_1$ is a minimization problem and and $\gamma = \frac{\alpha\beta\delta}{1-\alpha\beta\delta}$ if $\CP_1$ is a maximization problem.
\end{theorem}

	\section{Complete problems in {PPL}}\label{sec:ppl}

We have shown that $\PPL$ is closed under \DLOG reduction and defined $\PPL$-completeness in~\cite{DBLP:journals/tcs/GaoLML20}.
However, we did not manage to find the first natural $\PPL$-complete problem.
In this section, we give a negative answer to the existence of $\PPL$-complete problems.

\begin{lemma}\label{lma:PPL-close}
	{\rm \cite{DBLP:journals/tcs/GaoLML20}}
	For any two problems $\CP_1$ and $\CP_2$, if $\CP_2 \in \emph{\PPL}^i$, and there is a \emph{\DLOG} reduction from $\CP_1$ to $\CP_2$, then $\CP_1 \in \emph{\PPL}^{i + 1}$.
\end{lemma}

\begin{theorem}\label{thm:hierarchy}
	{\rm \cite{DBLP:journals/tcs/GaoLML20}}
	For any $i \in \mathbb{N}$, $\emph{\PPL}^i \subsetneq \emph{\PPL}^{i+1}$.
\end{theorem}

\begin{theorem}
	There is no {\rm\PPL}-complete problem under {\rm\DLOG} reduction.
\end{theorem}

\begin{proof}
	For contradiction, suppose there is a \PPL-complete problem $\CP$ under \DLOG reduction.
	Hence, there is a constant $c \ge 1$ such that $\CP \in \PPL^c$.
	For Theorem~\ref{thm:hierarchy}, for any $i \in \mathbb{N}$, there is a problem $\CP_{i + 1}$ which belongs to $\PPL^{i + 1}$ but not to $\PPL^i$.
	Let $k = c + 1$.
	Since $\CP$ is \PPL-complete, there is a \DLOG reduction from $\CP_{k + 1}$ to $\CP$.
	From Lemma \ref{lma:PPL-close}, it is derived that $\CP_{k + 1} \in \PPL^{c + 1} = \PPL^k$.
	This contradicts to the fact that $\CP_{k + 1} \in \PPL^{k + 1}\setminus \PPL^k$.
	\qed
\end{proof}

Notice that every un-trivial problems in $\PPL^1$ is $\PPL^1$-complete under \DLOG reduction.
It is still meaningful to find complete problems of each level in \PPL hierarchy.
	\section{Conclusion}\label{sec:con}

This paper studies the pseudo-sublinear-time reductions specialized for problems in big data computing.
Two concrete reductions $\le^{\tt PsT}_m$ and $\le^{\tt PsPL}_m$ are proposed.
It is proved that the complexity classes \Pproblem and \PsT are closed under $\le^{\tt PsT}_m$, and the complexity class \PsPL is closed under $\le^{\tt PsPL}_m$.
These provide powerful tools not only for designing pseudo-sublinear-time algorithms for some problems, but also for proving certain problems are infeasible in sublinear time after a \PTIME preprocessing.
More concretely, based on the fact that circuit value problem belongs to $\Pproblem\setminus \PsT$, the algebraic equation root problem is proved not in \PsT by establish a $\le^{\tt PsT}_m$ reduction from CVP to it.
Since CVP is \Pproblem-complete under \NC reduction, it may turn out to be an excellent starting point for many results, yielding pseudo-sublinear-time reductions for fundamental problems and giving unconditional pseudo-sublinear intractable results.
Then to separate \PsT and \PsPL, the \PsT-completeness is defined under $\le^{\tt PsPL}_m$.
We give out a range of possible \PsT-complete problems by proving that all of them are also \Pproblem-complete under \NC reduction.
We also extend the L-reduction to pseudo-sublinear time and prove it linearly preserves approximation ratio for pseudo-sublinear-time approximation algorithms.
Finally, we give an negative answer to the existence of \PPL-complete problems under \DLOG reduction.
This may guide the following efforts focusing on finding complete problems for each level of \PPL hierarchy.

	\section*{Acknowledgment}
	This work was supported by the National Natural Science Foundation of China under grants 61732003, 61832003, 61972110 and U1811461.
	
	\bibliographystyle{plain}
	\bibliography{COCOA_Reduction_2021_ref}
	
\end{document}